\documentclass[conference]{IEEEtran}
\usepackage[utf8]{inputenc} 
\usepackage[T1]{fontenc}
\usepackage{url}
\usepackage{ifthen}
\usepackage{cite}
\usepackage[cmex10]{amsmath} % Use the [cmex10] option to ensure complicance
\usepackage{amssymb}
\usepackage{amsthm}
\usepackage{mathtools}
\usepackage{tikz}
\usepackage{pgfplots}
\pgfplotsset{compat=1.18}
\usepackage{algorithm}
\usepackage{algpseudocode}

\newtheorem{theorem}{Theorem}

\renewcommand\baselinestretch{1}

\def\c{\mathbf{c}}
\def\m{\mathbf{m}}
\def\r{\mathbf{r}}
\def\v{\mathbf{v}}
\def\G{\mathbf{G}}
\def\b{\mathbf{b}}
\def\y{\mathbf{y}}

\def\BibTeX{{\rm B\kern-.05em{\sc i\kern-.025em b}\kern-.08em
    T\kern-.1667em\lower.7ex\hbox{E}\kern-.125emX}}
\begin{document}

\title{Characterization of Blind Code Rate Recovery in Linear Block Codes}

\author{\IEEEauthorblockN{Atreya Vedantam}
\IEEEauthorblockA{\textit{Department of Electrical Engineering} \\
\textit{Indian Institute of Technology Madras}\\
Chennai, India \\
ee22b004@smail.iitm.ac.in}
\and
\IEEEauthorblockN{Radha Krishna Ganti}
\IEEEauthorblockA{\textit{Department of Electrical Engineering} \\
\textit{Indian Institute of Technology Madras}\\
Chennai, India \\
rganti@ee.iitm.ac.in}
}

\maketitle

\begin{abstract} 
Forward Error Correction (FEC) is used ubiquitously in the communication pipeline. We explore noncooperative decoding where we aim to recover the code rate of a linear block code. We present a metric to characterize the quality of the code rate recovery which uses any rank based estimation technique. We derive a closed form expression for this metric in terms of the algorithmic and the environmental parameters and assert that it should be low for good recovery. We use this metric to derive an expression for a better code rate estimate in high noise conditions and compare it with existing estimates. Finally we validate the derived expression for the metric and the improved performance in the code rate estimate by simulating the recovery of a Low Density Parity Check (LDPC) code. This also enables us to derive the optimal algorithmic parameters for recovery. 
\end{abstract}


\section{Introduction}
Blind recovery of Forward Error Correction (FEC) code parameters has been of great interest for a long time and has  applications in Adaptive Modulation and Coding (AMC) \cite{amc}, cognitive radio systems \cite{crs} and spectrum surveillance \cite{military}. In this scenario, the  receiver does not have knowledge of the code parameters to decode the message. The receiver  must rely on statistical information obtained through a large number of coded frames.

Various algorithms have been proposed for such blind recovery in non-cooperative contexts. In \cite{valembois}, Valembois proved that the general problem of obtaining a low rank matrix is NP-complete and proposed a heuristic algorithm to find minimum weight code words (originally developed by Canteaut and Chabaud \cite{cantchab}). Extending this, Cluzeau \cite{cluzeau} and Barbier \cite{barbier} used rank reduction techniques to find dual code words and tune the parameters in their algorithms by striking a balance between a high detection probability and a high false alarm probability. However, all of the above methods constrain themselves to using  hard decoded bits and do not leverage symbol-level information (soft information). Sicot \cite{sicot} introduced the usage of information at the symbol-level to aid in the recovery of dual code words, but no  performance analysis has been provided.  Xia \cite{xiawu}, Moosavi \cite{moosavi} and Luwei \cite{luwei} used Log-Likelihood Ratio (LLR) based techniques to estimate the code. However, all of their methods assume a known set of candidate codes. Ramabadran \cite{ramabadran} used a rank based approach but fixed the recovery matrix aspect ratio without analysis. Wang \cite{wang} used symbol-level information based techniques, but did not quantify the optimality of the parameters in the algorithm.

This paper focuses on blind recovery of a $[n, k]$ linear block code. We use symbol-level information to filter and set up a matrix in a similar approach to Wang et al. Thereafter, we present a metric to characterize the quality of the code rate recovery. We argue why the metric is important and derive a closed form expression in terms of the algorithmic parameters and the environmental parameters. We use this metric to derive an expression for a better code rate estimate in high noise conditions and compare it with existing estimates. Finally, we validate the derived expression for the metric and the improved performance in the code rate estimate by simulating the recovery of an LDPC code in noisy conditions and derive the optimal parameters for code recovery. 

\section{Mathematical Model}
We assume that an information source generates independent and identically distributed  bits. These are divided into blocks of size $k$ and each block is called a message, denoted by $\m$. Each message is multiplied by an unknown (but same) matrix of size $k \times n$ called the generator matrix (denoted by $G$) to generate a code word, denoted by $\c$. Mathematically,
\begin{align*}
    \c = \m \G, \quad \c \in \{0,1\}^n, \m \in \{0, 1\}^k, \G \in \{0, 1\}^{k\times n}.
\end{align*}
All operations done on bits in this paper are done in $GF(2)$. We do not assume any form for the generator matrix (for e.g. that it is systematic). However $\G$ represents an injective map and hence has a rank of $k$. Therefore, exactly $k$ bits out of the $n$ bits in $\c$ are \textit{independent} and the other $n-k$ are linearly dependent on these $k$ bits. These independent bits in the code word $\c$ are called \textit{message bits} and the dependent bits are called \textit{parity check bits}. This is a crucial observation that we exploit in the code parameter recovery. The ratio $\rho = k/n$ is called the code rate. We assume that each bit in $c$ is equally likely to be $0$ or $1$. This is a reasonable assumption in practical settings. Each code word $\c$ is then modulated using Binary Phase Shift Keying (BPSK) to generate a symbol vector, denoted by
    $ \b = 1-2\c$.
 This vector (symbol by symbol) is then sent over an Additive White Gaussian Noise (AWGN) channel with an unknown Signal-to-Noise Ratio (SNR), and received as a vector denoted by $\r \in \mathbb{R}^n$. Hence,
\begin{align}\label{1}
    \r = \b + \v,
\end{align}
where $\v\in \mathbb{R}^{n} = [v_1,v_2,\hdots, v_n]$ is a vector of i.i.d. noise samples drawn  from a normal distribution with variance
$\sigma^2$ 
whose probability density function (PDF) is given by
\begin{align*}
    f(v_k) = \frac{1}{\sigma \sqrt{2\pi}} \exp \left(\frac{-v_k^2}{2\sigma^2}\right), ~~ k= 1,\hdots, n.
\end{align*}
$\sigma^2 = N_0/2$, where $N_0$ is the noise spectral density. The vector $\r$ also corresponds to the LLRs of the transmitted symbols and hence will be referred to as the LLR vector. The receiver performs hard decision decoding on this vector to generate the received bit vector $\y \in \{0, 1\}^n$, generated component wise as
\begin{align*}
    y_i = \begin{cases}
        1 & r_i < 0,\\
        0 & r_i \geq 0.
    \end{cases}
\end{align*}
In any communication system, multiple such code words are transmitted and received over a long time. We assume $M$ code words are transmitted. The $M$ code words denoted  by $\mathcal{C}=\{\c_1, \c_2, \ldots, \c_M\}$ are modulated using BPSK mapping and transmitted over the AWGN channel.   We denote the corresponding received symbol vectors as $\mathcal{R}=\{\r_1, \r_2, \ldots, \r_M\}$. In the remainder of this paper, $\Phi(\cdot)$ represents the cumulative distribution function (CDF) of the standard normal distribution 
defined as
\begin{align*}
    \Phi(x) = \frac{1}{\sqrt{2\pi}} \int_{-\infty}^x \exp \left( \frac{-u^2}{2} \right) du,
\end{align*}
 and $Q(\cdot)$ represents the standard Q-function.
 defined as
\begin{align*}
    Q(x) = \frac{1}{\sqrt{2\pi}} \int_x^\infty \exp\left(\frac{-u^2}{2}\right) \text{d}u.
\end{align*}
In this paper, we  assume that time synchronization has been achieved using techniques in \cite{cluzeau} and the length of the code word $n$ has been found using techniques in \cite{ramabadran}. The pilot bits, preambles, and frame headers are removed from the frames received to obtain code word frames at the receiver. We also assume that we have a sufficient number of code words $M$ for the recovery process, an amount we quantify later. The goal is to recover the code rate $\rho$ using $\mathcal{R}$ in a  robust manner. 

\section{Algorithm}
In high SNR settings (or if the code word length is low), rank estimation methods such as \cite{kuvaja} can be used to estimate the code rate in a blind setting. However, with noise and particularly for large block sizes, the existing  rank based methods fail since the number of errors in the received frames return a poor rank estimate. Brute force also becomes computationally infeasible. In this Section, we introduce a new robust algorithm  based on soft information to recover the code rate $\rho$ from observing $\mathcal{R}$.
The algorithm takes in input parameters $\{\mathcal{R}, t_1,t_2,n\}$ and outputs the code rate estimate $\hat{\rho}$.

 \begin{enumerate}
    
     \item Estimate the noise variance $\sigma^2$ using \begin{align*}
        \frac1M \sum_{i=1}^M\text{Var}(\r_i) = 1 + \hat{\sigma}^2.
    \end{align*}  This is because in  \eqref{1}, the noise vector is independent of the symbol vector and hence their variances add. Taking the average of this estimate on all messages gives us a good estimate of the noise variance.
    \item Estimate the SNR as 
    \begin{align*}
        \hat{\text{SNR}} = \hat{\sigma}^{-2}.
    \end{align*}
     Estimate $p_e$, the probability of a bit error as
$       p_e = Q\left(\sqrt{\hat{\text{SNR}}}\right)$.
    \item Iterate the following procedure for all $i \in [1, 2, \ldots , M]$. Assume $t_1 \in [0, 1]$ and $t_2 \in \mathbb{N}$ are two parameters of the algorithm. For each  $\r_i$,  consider its $k$-th element $r_i^k$. If $|r_i^k| < t_1$ mark the corresponding $r_i^k$ as \textit{unreliable}. Let $j_i$ be the total number of unreliable bits in $\r_i$. If $j_i < t_2$ mark the received vector $\r_i$ \textit{suitable}.
    \item Suppose $M_s$ received vectors are found suitable. Decode these suitable received vectors (by BPSK thresholding) and  arrange the decoded bit vectors (written as a row vector) one below the other to generate a matrix of size $M_s \times n$, called the \textit{word matrix} and denoted by $\mathbf{W}$. Observe that each column in $\mathbf{W}$ contains either all message bits or all parity check bits since all code words belong to the same code $\mathcal{C}$. These are called \textit{message columns} (denoted by $\mathbf{W}_m$) or \textit{parity check columns} (denoted by $\mathbf{W}_p$) respectively.
    \item Apply the Gaussian elimination algorithm on the word matrix $\mathbf{W}$ to obtain its Row Reduced Echelon Form (RREF) in the below form (this will require a column permutation),
    \begin{align*}
        \mathbf{R} = \begin{pmatrix}
    \mathbf{I}^{k \times k} & \mathbf{D}^{k \times n} \\
     \mathbf{0}^{(M_s-k)\times k} & \mathbf{0}^{(M_s-k)\times (n-k)}\\
\end{pmatrix},
    \end{align*}
    where $\mathbf{D}$ is a matrix that captures all the linear dependencies between the columns of the word matrix and $\mathbf{I}$ is the standard identity matrix.
    \item Find the mean of each column $i$ denoted by $\mu(R_i)$ in $\mathbf{R}$. 
   
     Let  $k'$ denote the number of columns of $\mathbf{R}$ for which $\mu(R_i) \leq 1/M_s$. 
The code rate estimate is then given by \begin{align}\label{coderateestimate}
        \hat{\rho} = \frac{k' - \mathbb{E}[C]}{n-\mathbb{E}[C]},
    \end{align}
    where $\mathbb{E}[C]$ is defined in \eqref{ec}.
An approximation for $\mathbb{E}[C]$ is provided in \eqref{ecapprox}.

\end{enumerate}

\section{Analysis}
In this Section, we provide intuition for the above algorithm and introduce a metric to analyze the quality of the code rate recovery. 
Observe that every parity check bit is a linear combination (modulo $2$) of some message bits. There are $n-k$ such linearly dependent sets. Hence, in a noiseless word matrix, there are $n-k$ parity check columns which depend on message columns. Suppose a subset of the message columns $\mathbf{w}_i \subseteq \mathbf{W}_m$ and one parity check column $p_i \in \mathbf{W}_p$ for some $i \in \{1, 2, \cdots,  n-k\}$ form a linearly dependent set. That is, there exist a binary vector $\mathbf{\lambda}_i \in \{0, 1\}^{k}$ and scalar $\alpha_i \in \{0,1\}$ such that $\mathbf{\lambda}_i \cdot \mathbf{w}_i + \alpha_i p_i = 0$ (where the `$\cdot$' operation is the usual inner product over $GF(2)$).
We represent this set by
\begin{align*}
    \mathcal{S}_i = \{\mathbf{w}_i, p_i\}, i \in \{1, 2, \cdots,  n-k\}.
\end{align*}
Now suppose that there was exactly one error in the entire word matrix, occurring in a column of parity check $p_i$. Only $\mathcal{S}_1$ is affected by this error and the rank estimate increases by one. On the other hand, suppose that this error occurred in a message column $w_i$. This would affect all subsets $\mathcal{S}_j$ where $w_i \in \mathcal{S}_j$, possibly leading to increases in the multiple rank estimate. Therefore, for a good recovery, multiple columns should have no bit error.

We claim that the number of columns in the word matrix which have at least one bit error characterizes the recovery quality. To show this, we show that the presence of at least one bit error in a column will yield an increased rank estimate with high probability.

We consider a toy model with $d$ message columns and one parity check column linearly dependent on all the message columns. From the above discussion, a single error will increase the rank estimate from $d$ to $d+1$. Conditioned on at least one bit error in the word matrix, we obtain the correct rank estimate only when the bit errors retain the linear dependency. In other words, the erroneous parity check column must remain in the span of the erroneous message columns. We show that this probability is vanishingly small for large word matrices.

\begin{theorem}
    Suppose $p, m_1, m_2, \ldots, m_d \in \{0, 1\}^{M_s}$ are columns of the word matrix such that $p = \bigoplus_{i=1}^d m_i$. Let $p_e'$ be the probability of a bit error in the code word matrix. Suppose $M_s >> d$ and $p_e'$ was small such that $M_sp_e'$ is large. Then the probability that the rank increases by $1$ in the presence of at least one error (denoted by $p_1$) is lower bounded as follows.
    \begin{align}\label{rankincrease}
        p_1 \geq \Phi \left(\frac{(M_s-d)^2 - 2M_s^2(d+1)p_e'}{2M_s(d+1)\sqrt{M_sp_e'(1-p_e')}}\right).
    \end{align}
\end{theorem}
\begin{proof}
    See Appendix for a detailed proof. 
    Proof sketch: We (almost surely) lower bound the minimum distance of the code formed by the span of $\{ m_1, m_2, \ldots, m_d\}$. This distance increases with increasing $M_s$. Now given at least one bit error in the matrix, there must be enough errors to `cross' this minimum distance, if the linear dependency is to be retained. For a small bit error probability and high $M_s$, this becomes vanishingly unlikely, concluding the proof.
\end{proof}

Simplifying  \eqref{rankincrease} for $M_s >> d$ and assuming $p_e' << 1/(d+1)$ we have
\begin{align*}
    p_1 \approx \Phi \left(\frac{\sqrt{M_s}}{2(d+1)\sqrt{p_e'(1-p_e')}}\right).
\end{align*}
For reasonable choices like $M_s = 10^3$, $d=10$, $p_e' = 0.02$, $p_1$ is almost $1$. Therefore when the noise is high enough such that $M_sp_e'$ is large, the code recovery will not be accurate. This motivates why the number of columns with at least one bit error captures the idea of a good code recovery quality.

Now, we analytically derive the expected number of columns with at least one bit error.

\begin{theorem}
    Let $p_e$ denote the probability of a bit error. Let $p_u$ denote the probability of a bit being unreliable.
    \begin{itemize}
        \item[(1)] Then $p_u$ depends on the parameter $t_1$ as \begin{align}\label{pu}
            p_u = Q\left( \frac{1-t_1}{\sigma}\right)-Q\left( \frac{1+t_1}{\sigma}\right). 
                \end{align}
        \item[(2)] If $F(k; n, p) = \sum_{i=0}^k {n \choose i} p^i (1-p)^{n-i}$ represents the Binomial CDF, then the expected number of columns with at least one error is given by
        \begin{align}\label{ec}
        \mathbb{E}[C] \approx n-n\left(1-\frac{p_e F(t_2-1; n-1, p_u)}{F(t_2, n, p_u)}\right)^{M_s}.
        \end{align}
        When $p_e$ is small enough such that $M_sp_e << 1$, we have
        \begin{align}\label{ecapprox}
           \mathbb{E}[C] = n\mathbb{E}[M]p_e F(t_2-1; n-1, p_u),
        \end{align} where $\mathbb{E}[M]$ is the total expected number of noisy codewords needed to obtain $M_s$ suitable codewords. This is given by
        \begin{align}\label{expmsgs}
    \mathbb{E}[M] = \frac{M_s}{F(t_2; n, p_u)}.
\end{align}
    \end{itemize}
\end{theorem}

\begin{proof}
    See Appendix for a detailed proof.
    Proof sketch: The first part follows from integrating the normal PDF over the `unreliable' zone. For the second, $\mathbb{E}[C]$ is proportional to the probability that one column has at least one bit error. This probability depends on the expected number of bit errors in each code word. This quantity can be calculated using the distribution of the number of unreliable and reliable bits in each code word and its effect on the number of bit errors.
\end{proof}

\begin{figure}
    \centering
    \includegraphics[width=0.8\linewidth]{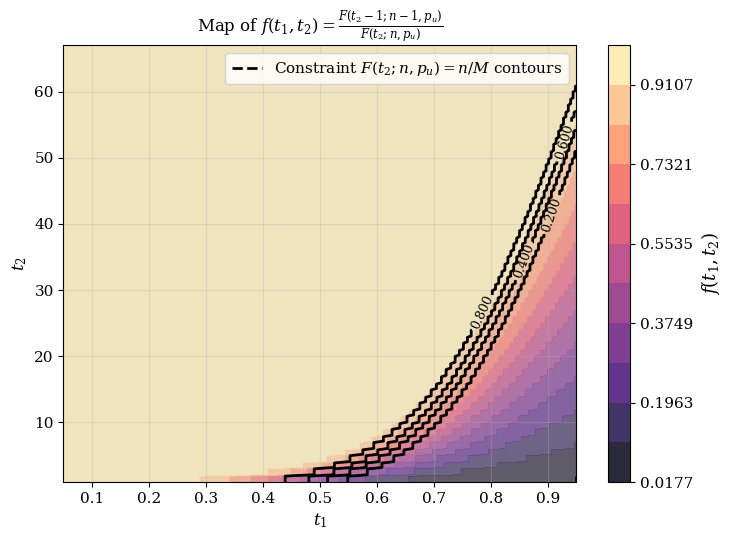}
    \caption{Contour plot of the algorithmic error versus $t_1$ and $t_2$ for a code of length $n = 136$ received at an SNR of $10$ dB. $t_1^* = 0.99$ and $t_2^* = 1$ and the minimum $f(t_1^*, t_2^*) = 0.0153$. Black lines represent constant $M$ curves (see equation \eqref{contours}).}
    \label{fig:contours}
\end{figure}

Equation \eqref{ecapprox} has intuitive  meaning. We observe that $nM_s$ is the total number of bits in the word matrix. Multiplying this by $p_e$ gives us the fraction of these bits in error due to the environment. Hence, we call the term $nM_sp_e$ the \textit{ambient} error. The other term is the factor by which algorithmic optimization helps decrease the error. We call this term the \textit{algorithmic} error. The optimal $t_1$ and $t_2$ values minimize $\mathbb{E}[C]$. Therefore the optimal values minimize the algorithmic error.
\begin{align*}
    t_1^*, t_2^* = \arg \min_{t_1, t_2} \frac{F(t_2-1; n-1, p_u)}{F(t_2; n, p_u)}.
\end{align*}
It is important to note that this also places a (rather stringent) requirement on the expected number of recovery messages (see \eqref{expmsgs}). This can be seen in Fig. \eqref{fig:contours}. 

\begin{figure}
    \centering
    \includegraphics[width=0.85\linewidth]{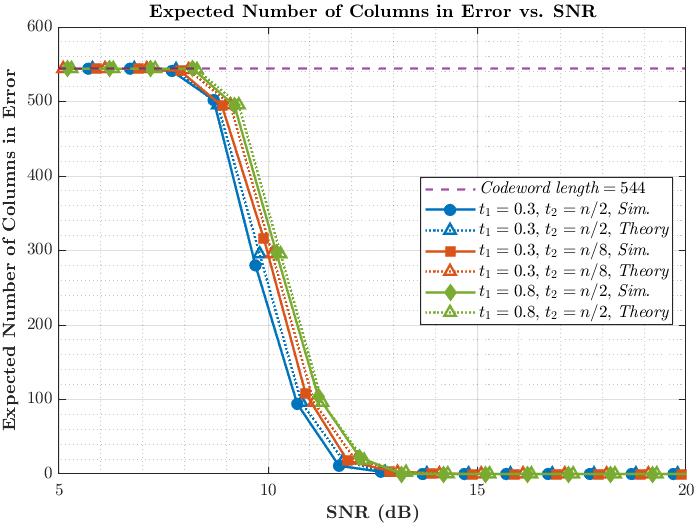}
    \caption{Theoretical and Simulated $\mathbb{E}[C]$ for three cases: $t_1 = 0.3, t_2=n/2$, $t_1 = 0.3, t_2 = n/8$, $t_1 = 0.8, t_2 = n/2$ for $1,000$ messages of a $544$ length LDPC code from $5$ to $20$ dB SNR. The close match demonstrates the correctness of the theoretical result over a wide range of $t_1$ and $t_2$.}
    \label{fig:simultheory}
\end{figure}

Fig. \eqref{fig:contours} shows a contour plot of the algorithmic error as a function of $t_1$ and $t_2$ for a code of length $n = 136$ at an SNR of $10$ dB. The four black lines represent the contours of the function $F(t_2; n, p_u)$ when it is equal to $0.2$, $0.4$, $0.6$ and $0.8$. The lower values correspond to a higher (expected) message requirement, as seen in \eqref{expmsgs}. The plot tells us that the random choice of $t_1$ and $t_2$ is harmful because the majority of the landscape yields a high algorithmic error. We also observe that the contours of the objective function are almost parallel to the contours of constant $F(t_2; n, p_u)$. Therefore, minimizing algorithmic error also greatly reduces $F(t_2; n, p_u)$. This implies that we need to observe many more messages to find $M_s$ suitable code words. In some scenarios, this might not be possible. We reformulate the problem when $M$ becomes a fixed value. 

If $M_s < n$, the rank of the word matrix would be limited by $M_s$. Since we do not want that, it must follow that $M_s \geq n$. For large $M_s$, $\mathbb{E}[C]$ becomes large as well since each column is longer and is more susceptible to bit errors. Therefore, in order to balance the limitation of rank estimation and more errors, it is optimal to set $M_s = n$. Then \eqref{expmsgs} places a constraint on what values $t_1$ and $t_2$ can take. In this case, the optimal values are given by
\begin{align}
    t_1^*, t_2^* =  \arg \min_{t_1^*, t_2^*} F(t_2-1; n-1, p_u) \\ \label{contours}
    \textrm{s.t. } F(t_2; n, p_u) = \frac{n}{M}.
\end{align}
Since the contours are almost parallel, not much optimization can be done along the constraint. This leads to the intuitive conclusion that the ambient error is the more significant limiting factor in recovery.

Suppose $k'$ is the observed rank of the word matrix. Clearly $k' > k$. The difference $k'-k$ represents the number of columns in the linearly dependent subset of the word matrix that have stopped showing that relationship due to bit errors. On average, a $1-k/n$ fraction of $\mathbb{E}[C]$ columns having at least one bit error are parity check columns. Each of these certainly contributes to one rank increase. Hence we can say
\begin{align*}
    k' - k \geq \left(1-\frac{k}{n}\right) \mathbb{E}[C].
\end{align*}
Rearranging gives us an estimate of the code rate estimate that is used in equation \eqref{coderateestimate} in the algorithm.
\begin{align}\label{cresttheory}
    \hat{\rho} = \frac{k}{n} \approx \frac{k'-\mathbb{E}[C]}{n - \mathbb{E}[C]}.
\end{align}

Equation \eqref{cresttheory} can also be intuitively interpreted. If we imagine removing $\mathbb{E}[C]$ columns from the estimated message columns, the code rate of this new code is a better estimate.

\section{Simulation Results}
In this Section, we show that the theoretical estimate of $\mathbb{E}[C]$ matches closely with observed values in practical noise scenarios. Subsequently, we show that the proposed correction for the code rate estimate in the algorithm is good for a good range of SNR.

\begin{figure}
    \centering
    \includegraphics[width=0.85\linewidth]{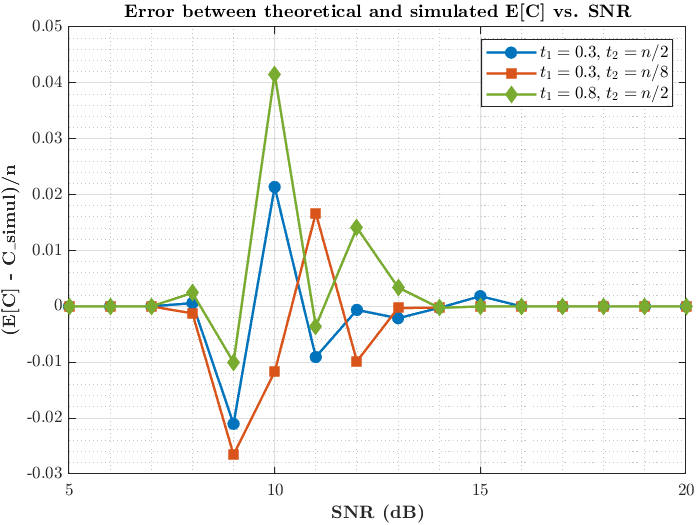}
    \caption{Error between the theoretical $\mathbb{E}[C]$ and observed number of columns in error normalized to the code length ($544$). The error goes to $0$ for high SNR.}
    \label{fig:error}
\end{figure}

\begin{figure}
    \centering
    \includegraphics[width=0.85\linewidth]{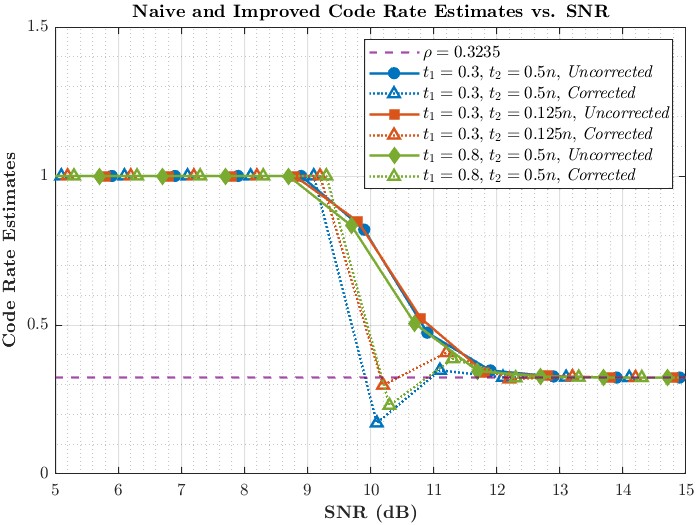}
    \caption{The naive code rate estimate $k'/n$ (solid lines) and improved estimate $(k'-\mathbb{E}[C])/(n - \mathbb{E}[C])$ (dashed lines) versus SNR for three different cases of $t_1$ and $t_2$. Simulation is done for $1000$ messages of a $544$ length LDPC code with a true code rate of $0.3235$.}
    \label{fig:coderateest}
\end{figure}

We used MATLAB to simulate our experiments. We generate $1000$ random binary messages and code them using an LDPC code of length $544$. We used standard LDPC coding for 5G New Radio (NR). The code is not punctured and its true rate is $\rho = 0.3235$ (close to $1/3$). Each code word is modulated with BPSK and sent over channels with SNRs varying from $5$ dB to $20$ dB. The recovery algorithm described in Section $2$ is used. 

The formula for $\mathbb{E}[C]$ is first validated for a wide range of $t_1$ and $t_2$ and SNRs. This is shown in Fig. \ref{fig:simultheory}. The error is less than $5\%$ as seen in Fig. \eqref{fig:error}. 

In Fig. \ref{fig:coderateest} the naive code rate estimate is shown along with the proposed correction for various SNRs. The points in the figure are slightly laterally shifted to prevent overlap for clarity. For $5$ to $9$ dB SNRs the word matrix becomes full rank due to errors and it is trivial to see that the proposed correction does not help ($k'=n \implies k'-\mathbb{E}[C] = n - \mathbb{E}[C]$). At $10$ dB SNR, the word matrix does not become full rank and the proposed correction sharply brings the estimated code rate close to the true rate $0.3235$. The naive code rate estimate converges to the true code rate only at $12$ dB in this case. This empirically shows that the correction enables us to accurately estimate the code rate $2$ dB lower than the naive estimate. We also see that the accuracy of the code rate estimate exactly follows the $\mathbb{E}[C]$ curve from Fig. \eqref{fig:simultheory}. This verifies that $\mathbb{E}[C]$ characterizes the code rate estimate quality. 

Below we compare this corrected estimate with existing methods:    1) Ramabadran \cite{ramabadran} simulates a [648, 162] LDPC code with BPSK modulation and recovers the code rate at SNRs $\geq 10$ dB, matching the performance of our method. The number of code words used was $1296$, similar to our simulation.
    2)  Wang \cite{wang} simulates a [648, 540] LDPC code and recovers the code at a BER of $1.5 \times 10^{-3}$ that in a modulated stream BPSK corresponds to a $8.5$ dB SNR. However, the recovery process required $3,840,000$ code words, whereas our proposed algorithm recovers a code of similar size using only $1000$ code words.

We also observe that the code rate estimate at $10$ dB SNR is poorer when $t_1 = 0.3$ and $t_2 = n/2$ compared to the other choices. This highlights the importance of optimizing $t_1$ and $t_2$. At $10$ dB SNR, $t_1 = 0.07$ and $t_2 = 0.007$ are found to be close to the optimal parameters obtained by solving \eqref{contours}. This was found using a grid search of the parameters $t_1$ and $t_2$ over $[0, 1]$. The optimization landscape is smooth so the grid search yields close to optimal values for reasonable resolution.

For higher block lengths, the recovery becomes harder because of the increased word matrix size. We generated $2000$ random binary messages of LDPC code of length $1088$ with the same code rate $\rho = 0.3235$. The results for this simulation are shown in Fig. \ref{fig:coderateest_1088}. We see that the proposed correction sharply brings the naive estimate to the true value when the word matrix is rank deficient (at SNRs $\geq 11$ dB). 

For very large block lengths of order $10^4$, it is difficult in practice to obtain more than (order of) $10^4$ code words. Then the ambient error is $nM_sp_e = 10^8p_e$ which should be small compared to $10^4$. Therefore, in practice, we need $p_e \approx 10^{-6}$ for a good recovery and this requires an SNR of at least $10$ dB.

\begin{figure}
    \centering
    \includegraphics[width=0.85\linewidth]{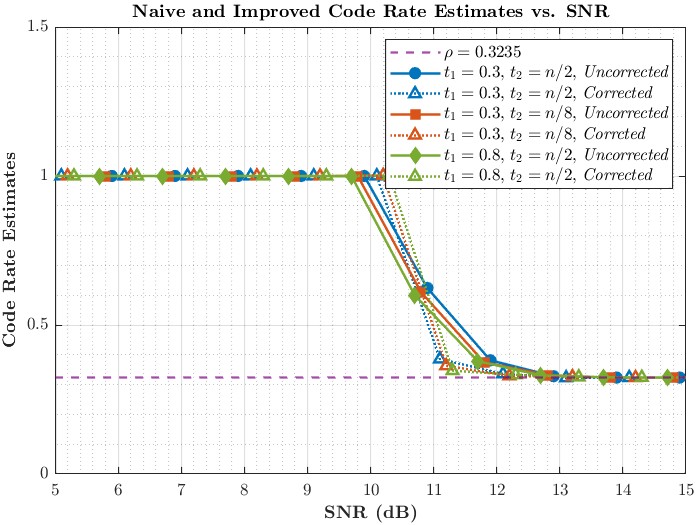}
    \caption{The naive code rate estimate $k'/n$ (solid lines) and improved estimate $(k'-\mathbb{E}[C])/(n - \mathbb{E}[C])$ (dashed lines) versus SNR for three different cases of $t_1$ and $t_2$. Simulation is done for $2000$ messages of a $1088$ length LDPC code with a true code rate of $0.3235$.}
    \label{fig:coderateest_1088}
\end{figure}

\section{Conclusion}
We proposed a metric ($\mathbb{E}[C]$) to characterize the quality of the blind recovery of code parameters in any rank based estimation technique in linear block codes. We motivated the reasoning behind using the metric by arguing in the presence of at least one error, the rank of the word matrix increases by one with high probability. We then derived a closed form expression for this metric which depended on two factors: the ambient error and the algorithmic error. 

We then used this metric to derive a better code rate estimate. Finally we minimized the algorithmic error to derive optimal parameter values for recovery for two practical constraints: one when the receiver can control the number of noisy code words received and the other when they cannot. We verified each of these theoretical results by simulating a $544$ length LDPC code over various SNRs and comparing it with existing similar approaches.

We discussed a practical lower bound for the SNR for the reasonable recovery of the codes with large code word lengths.

\bibliographystyle{IEEEtran}
\bibliography{definitions,bibliofile}

\appendices

\section{Proof of Theorem 1}
\begin{proof}
    We first observe that the column vectors $m_1, m_2, \ldots, m_d$ are random vectors. The span of these vectors form a random linear code $S$. The minimum distance of this code $d_m$ is almost surely bounded by 
    \begin{align*}
        H_2\left(\frac{d_m}{M_s}\right) \geq 1 - \frac{k}{M_s},
    \end{align*}
    where $H_2(\cdot)$ is the binary entropy function \cite{barg}. Therefore we have 
    \begin{align*}
        2\sqrt{\frac{d_m}{M_s}\left(1-\frac{d_m}{M_s}\right)} \geq 1 - \frac{k}{M_s},
    \end{align*} and neglecting $d_m/M_s$ relative to $1$, we obtain 
    \begin{align}\label{mindist}
        d_m \geq \frac{(M_s-k)^2}{2M_s}.
    \end{align}
    Now suppose each $m_i$ and $p$ are corrupted with noise vectors $e_i$ and $e$ respectively. Hence $m_i' = m_i \oplus e_i$ and $p' = p \oplus e$. A noiseless $p$ lies inside $S$. To preserve the rank, $p'$ must lie in the span of $m_1', m_2', m_3', \ldots, m_d'$. In the worst case scenario, $p' = p \oplus e = \bigoplus_{i=1}^d m_i' = \bigoplus_{i=1}^d (m_i \oplus e_i)$ which implies $p \oplus (e \bigoplus_{i=1}^d e_i)$ lies in the span of $m_1, m_2, m_3, \ldots, m_d$. If the total weight of the sum of the above $d+1$ error vectors does not exceed $d_m$, the erroneous vector $p'$ cannot be in the new span and hence the rank must increase by $1$. Therefore we have
    \begin{align*}
        p_1 \geq \mathbb{P}\left(\left|e \bigoplus_{i=1}^d e_i\right| \leq d_m \right).
    \end{align*} 
 
 The above is lower bounded when none of the error vectors have intersecting supports and hence
 \begin{align*}
        p_1 \geq \mathbb{P}\left(|e| \leq \frac{d_m}{d+1} \right).
    \end{align*}  
    Since $e$ is a vector of Bernoulli distributed ones, it follows that $|e|$ is Binomially distributed with mean $M_s p_e'$ and variance $M_sp_e'(1-p_e')$. For $M_s >> d$ this can be approximated by a normal CDF (after substituting for $d_m$ using \eqref{mindist}). This concludes the proof.
\end{proof}

\section{Proof of Theorem 2}
\begin{proof}
The probability that a bit is unreliable is independent of whether it is a $0$ or a $1$. Therefore $p_u$ is simply given by
\begin{align*}
    p_u = \int_{1-t_1}^{1+t_1}  \frac{1}{\sigma \sqrt{2\pi}} \exp \left(\frac{-\nu^2}{2\sigma^2}\right) d\nu,
\end{align*}
which reduces to equation \eqref{pu} trivially.
We define a few random variables to aid our proof of equation \eqref{ec}. All the random variables defined below are conditioned on the code words being chosen from the word matrix and hence have $t_2$ or lesser unreliable bits. For each code word, let $U_k$ and $R_k$ be random variables denoting the number of bit errors among the given $k$ unreliable bits and $n-k$ reliable bits respectively. Let $E_k$ denote the number of bit errors in a code word conditioned on it having $k$ unreliable bits. Then we have
\begin{align*}
    E_k = U_k + R_k.
\end{align*}
Let $C_i$ be the random variable which takes the value $1$ if the $i$th column of the word matrix has at least one bit error and $0$ otherwise. Let $C$ be the random variable denoting the number of columns of the word matrix having at least one bit error. We then have 
    $C = \sum_{i=1}^n C_i$.
Let $E$ be the random variables denoting the total number of bit errors in a code word and $K$ be the random variable denoting the number of unreliable bits in a code word. Let $M$ denote the number of code words required until we have enough suitable code words for recovery. Hence from the identical distribution of $C_i$,
\begin{align}\label{ecog}
    \mathbb{E}[C] = \sum_{i=1}^n \mathbb{E}[C_i] = n\mathbb{E}[C_1] = n\mathbb{P}[C_1=1].
\end{align}
From the total expectation theorem,
\begin{align*}
    \mathbb{P}[C_1=0] = \left(\displaystyle \sum_{e=0}^n \left( 1-\frac{e}{n}\right) \mathbb{P}[E = e]\right)^{M_s}
\end{align*}
\begin{align}\label{pc1is0}
    \implies \mathbb{P}[C_1=0] = \left(\displaystyle 1-\frac{\mathbb{E}[E]}{n}\right)^{M_s}.
\end{align}
Now $\mathbb{E}[E]$ can be written as
\begin{align*}
    \mathbb{E}[E] &= \sum_{k=0}^n \mathbb{E}[E_k] \mathbb{P}[K = k] 
    = \sum_{k=0}^n (\mathbb{E}[U_k]+\mathbb{E}[R_k)] \mathbb{P}[K = k].
\end{align*}
Let $p_{e|u}$ denote the probability of an unreliable bit in error and $p_{e|r}$ denote the probability of a reliable bit being in error. Then $U_k$ and $R_k$ are binomially distributed as $U_k \sim \text{Binomial}(k, p_{e|u})$ and $R_k \sim \text{Binomial}(n-k, p_{e|r})$. Substituting for their expected values,
\begin{align*}
    \mathbb{E}[E] &= \sum_{k=0}^n (kp_{e|u} + (n-k)p_{e|r})\mathbb{P}[K = k].
\end{align*}
$K$ is distributed as a truncated Binomial (since we only consider the word matrix). Hence the Probability Mass Function (PMF) of $K$ is
\begin{align*} \mathbb{P}[K = k] = \begin{dcases}
\dfrac{\binom{n}{k}p_u^k(1-p_u)^{n-k}}{F(t_2; n, p_u)} , k \leq t_2\\
0, k > t_2.
\end{dcases}
\end{align*}
Hence the sum reduces (after some manipulation) to
\begin{align*}
    &=\sum_{k=0}^{t_2} (kp_{e|u} + (n-k)p_{e|r}) \left\{ \frac{\binom{n}{k}p_u^k(1-p_u)^{n-k}}{F(t_2; n, p_u)} \right\},\\
    &=\frac{np_{e|u}p_u}{F(t_2; n, p_u)}\sum_{k=1}^{t_2} \binom{n-1}{k-1}p_u^{k-1}(1-p_u)^{n-k} \\
    &+
    \frac{np_{e|r}(1-p_u)}{F(t_2; n, p_u)}\sum_{k=0}^{t_2} \binom{n-1}{k}p_u^k(1-p_u)^{n-k-1}.
\end{align*}
Both the terms above are binomial CDFs. From our notation, the above expression reduces to
\begin{align*}
    =\frac{np_{e|u}p_u F(t_2-1; n-1, p_u) + np_{e|r}p_r F(t_2; n-1, p_u)}{F(t_2; n, p_u)}.
\end{align*}
We observe that $p_e = p_{e|u}p_u+p_{e|r}p_r$ and this nicely comes out as a common factor. The remaining term is negligible because $p_{e|r}p_r$ is small compared to $p_{e|u}p_u$ and $\binom{n-1}{t_2} p_u^{t_2} (1-p_u)^{1-t_2}$ is small compared to $F(t_2-1; n-1, p_u)$ for reasonably large $t_2$. Hence we have
\begin{align*}
    \mathbb{E}[E] \approx \frac{np_e F(t_2-1; n-1, p_u)}{F(t_2; n, p_u)}.
\end{align*}
Substituting this back in \eqref{pc1is0} gives us
\begin{align*}
    \mathbb{P}[C_1=0] \approx \left(\displaystyle 1-\frac{p_e F(t_2-1; n-1, p_u)}{F(t_2; n, p_u)}\right)^{M_s}.
\end{align*}
Substituting this in equation \eqref{ecog} gives us equation \eqref{ec}. If $p_e$ is low (such that $M_sp_e$ is small), then this reduces to 
\begin{align}\label{ratiobinom}
    \mathbb{E}[C] = \frac{nM_sp_e F(t_2-1; n-1, p_u)}{F(t_2; n, p_u)}.
\end{align}
% If $M_s < n$, the rank of the word matrix would be limited by $M_s$. Since we do not want that, it must follow that $M_s \geq n$. For large code word lengths, this requires a large amount of frames to be stored and hence it is practically reasonable to set $M_s = n$. 

$M$ follows a negative binomial distribution with a success probability $F(t_2; n, p_u)$. Therefore the expected number of messages required is given by
\begin{align}\label{expmsgsproof}
    \mathbb{E}[M] = \frac{M_s}{F(t_2; n, p_u)},
\end{align}
and substituting this in equation \eqref{ratiobinom} completes the proof.
\end{proof}
 
 \end{document}